\documentclass[11pt]{article}
\usepackage{titlesec}
\usepackage[toc,page]{appendix}
\usepackage{amsmath}

\usepackage{amsthm}
\usepackage{algorithm}
\usepackage{algorithmicx}
\usepackage[noend]{algpseudocode}
\usepackage[symbol]{footmisc}
\renewcommand{\thefootnote}{\fnsymbol{footnote}}
\usepackage{tikz}
\usepackage{multicol}
\usepackage{mathtools}
\usepackage{adjustbox}
\usepackage{diagbox}
\usepackage{array}
\newcolumntype{P}[1]{>{\centering\arraybackslash}p{#1}}
\newcolumntype{M}[1]{>{\centering\arraybackslash}m{#1}}
\DeclarePairedDelimiter{\ceil}{\lceil}{\rceil}
\makeatletter
\def\BState{\State\hskip-\ALG@thistlm}
\makeatother
\usetikzlibrary{calc, chains,
                fit,
                positioning,
                shapes}

\definecolor{myblue}{RGB}{80,80,160}
\definecolor{mygreen}{RGB}{80,160,80}
\definecolor{myempty}{RGB}{255,255,255}

\usepackage{caption}

\usepackage{graphicx}
\usepackage{hyperref}
\hypersetup{
    bookmarks=true,         
    unicode=false,          
    pdftoolbar=true,        
    pdfmenubar=true,        
    pdffitwindow=false,     
    pdftitle={Certificate},    
    pdfauthor={Gal Oren, Leonid Barenboim},     
    pdfsubject={Fast Distributed Backup Placement in Sparse and Dense Networks},   
    pdfcreator={Gal Oren, Leonid Barenboim},   
    pdfproducer={},  
    pdfkeywords={},
    pdfnewwindow=true,      
    colorlinks=false,       
    linkcolor=red,          
    citecolor=green,        
    filecolor=magenta,      
    urlcolor=cyan           
}
\usepackage[utf8]{inputenc}
\usepackage[english]{babel}
\newtheorem{theorem}{Theorem}[section]
\newtheorem{corollary}{Corollary}[theorem]
\newtheorem{lemma}[theorem]{Lemma}

\def\denseformat{
	\setlength{\textheight}{9in}
	\setlength{\textwidth}{6.9in}
	\setlength{\evensidemargin}{-0.2in}
	\setlength{\oddsidemargin}{-0.2in}
	\setlength{\headsep}{10pt}
	\setlength{\topmargin}{-0.3in}
	\setlength{\columnsep}{0.375in}
	\setlength{\itemsep}{0pt}
}

\denseformat

\titlespacing*{\section}
{0pt}{2ex}{0.5ex}
\titlespacing*{\subsection}
{0pt}{0ex}{0.5ex}

\begin{document}
\title{Fast Distributed Backup Placement in Sparse and Dense Networks}


\author{Leonid Barenboim \thanks{Open University of Israel. E-mail: 
        {\tt\small leonidb@openu.ac.il}} \  , \ Gal Oren \thanks{Ben-Gurion University of the Negev, Nuclear Research Center - Negev. E-mail:
        {\tt\small orenw@post.bgu.ac.il} \newline This work was supported by the Lynn and William Frankel Center for Computer Science, the Open University of Israel's Research Fund, and ISF grant 724/15.}
}





\date{}

\def\thepage{}
\maketitle 
\vspace{-0.4cm}
\begin{abstract}
We consider the Backup Placement problem in networks in the $\mathcal{CONGEST}$ distributed setting. Given a network graph $G = (V,E)$, the goal of each vertex $v \in V$ is selecting a neighbor, such that the maximum number of vertices in $V$ that select the same vertex is minimized. The backup placement problem was introduced by Halldorsson, Kohler, Patt-Shamir, and Rawitz \cite{halldorsson2015bp}, who obtained an $O(\log n/ \log \log n)$ approximation with randomized polylogarithmic time. Their algorithm remained the state-of-the-art for general graphs, as well as specific graph topologies.
In this paper we obtain significantly improved algorithms for various graph topologies. Specifically, we show that $O(1)$-approximation to optimal backup placement can be computed {\em deterministically} in $O(1)$ rounds in graphs that model wireless networks, certain social networks, claw-free graphs, and more generally, in any graph with {\em neighborhood independence bounded by a constant}. At the other end, we consider sparse graphs, such as trees, forests, planar graphs and graphs of constant arboricity, and obtain a constant approximation to optimal backup placement in $O(\log n)$ deterministic rounds. 

Clearly, our constant-time algorithms for graphs with constant neighborhood independence are asymptotically optimal. Moreover, we show that our algorithms for sparse graphs are not far from optimal as well, by proving several lower bounds. Specifically, optimal backup placement of unoriented trees requires $\Omega(\log n)$ time, and approximate backup placement with a polylogarithmic approximation factor requires $\Omega(\sqrt {\log n / \log \log n})$ time.
Our results extend the knowledge regarding the question of "what can be computed locally?" \cite{naor1995can}, and reveal surprising gaps between complexities of distributed symmetry breaking problems.

\end{abstract}
\section{Introduction}
\label{intro1}

We consider the Backup Placement problem in networks. The problem is defined as follows. Given a network graph $G = (V,E)$, the goal of each vertex $v \in V$ is selecting a neighbor, such that the maximum number of vertices in $V$ that select the same vertex is minimized.  The number of neighbors that select a certain vertex $v \in V$ is called the {\em load} on $v$. The load of a backup placement solution for $G= (V,E)$ is the maximum load on  a vertex $v \in V$. For a positive integer $t \geq 1$, {\em  $t$-backup placement} for $G$ is a solution in which each vertex $v \in V$ selects a neighbor, such that the maximum load in $V$ is at most $t$ times the load of an optimal solution for $G$. 
To illustrate this problem, consider the following two simple examples. First, suppose that $G$ is a cycle graph. (See Figure \ref{fig1} in Appendix A). In this case, the optimal solution to the backup placement problem is when each node selects its succeeding neighbor to be its backup node (in blue in Figure \ref{fig1}). By that, the maximal load in this graph is of only one unit. Next, suppose that $G$ is a rooted tree as in Figure \ref{fig1} (right). The optimal solution in this case forces all the nodes to choose their parents (besides the root) to be their backup nodes. By that, the maximal load is $\Delta = \Delta(G)$, the maximum degree in the graph. Similarly, in any regular tree, no solution with load smaller than $\Delta - 1$ exists. On the other hand, some non-regular trees admit very good solutions, of load as small as $2$ units.


Finding a backup placement in a network while minimizing the load on the network vertices is a very important goal \cite{halldorsson2015bp,halldorsson2018distributed,oren2018distributed}. First of all, it allows each vertex to perform a backup to a neighboring node, rather than a more distant destination, and thus improves network performance. In addition, nodes' memories are used to the minimum extent for the purpose of backups, which makes it possible to maximize the memory available for other purposes of the vertices. Moreover, if a certain node fails, the number of backups that become unavailable is minimized. For these and other reasons, the backup placement problem is considered as a central distributed problem by Halldorsson, Kohler, Patt-Shamir and Rawitz \cite{halldorsson2015bp},  who initiated the study of this problem in the distributed setting in 2015.

We would like to point out an additional important aspect. This problem is also a central symmetry-breaking problem, and is interesting to compare to such important problems as Maximal Matching and Maximal Independent Set. In a certain sense the backup placement problem generalizes Maximal Matching. Note that a perfect matching is also a perfect backup placement. Indeed, all vertices can perform backups to their pairs in the matching. When a perfect matching is not possible, then an optimal backup placement assigns up to $c$ selecting neighbors to each vertex, for a parameter $c > 0$. Therefore, this generalizes Maximal Matching, in which $c = 1$, but certain vertices may be selected by no vertices at all. Since Maximal Matching and other symmetry-breaking problems cannot be solved in a constant number of rounds \cite{linial1992locality,kuhn2004cannot}, even on oriented trees and unit-disk graphs, the question of whether (approximate) backup placement can be solved in various networks within this number of rounds is of great interest.

\renewcommand{\thefootnote}{\arabic{footnote}}
In this paper we answer this question in the affirmative. Specifically, we show that $O(1)$-backup-placement can be computed in $O(1)$ rounds in  wireless networks, certain social networks, claw-free graphs, and more generally, in any graph with {\em neighborhood independence bounded by a constant.}\footnote[1]{ {\em Neighborhood independence} is the maximum number of independent neighbors a vertex in the graph has. 
}
At the other end, we consider sparse graphs, such as trees, forests, planar graphs and graphs of constant arboricity\footnote[2]{ {\em Arboricity} is the minimum number of forests that the graph edges can be partitioned into.}, and obtain a constant approximation to backup placement in $O(\log n)$ rounds. These results suggest that the problem is harder on sparse graphs than on dense ones. We prove this formally, by providing a lower bound for computing (approximate) backup placement on unoriented trees of $\Omega(\sqrt{\log n / \log \log n})$ rounds. Our lower bound holds for approximation ratios up to polylogarithmic in $n$. This is in contrast to oriented trees, where a constant approximation to backup placement is achieved within $O(1)$ rounds, and various dense graphs in which it is possible to compute a sparse infrastructure based on oriented trees within a constant number of rounds. All our algorithms work in the $\mathcal{CONGEST}$ model, where only $O(\log n)$ bits are sent per edge per round. Moreover, our algorithms are based on relatively simple internal computations, which makes them especially appropriate for implementation in networks with limited resources, such as sensor networks and Internet of Things. See Table 1 for a comparison of our results with various symmetry-breaking results in various graph families. 

Our results reveal interesting separations between complexities of several problems and several graph topologies, as one can see in Table 1. On the other hand, for certain topologies that are usually separated, our algorithms achieve the same performance of $O(1)$ rounds. We note that network topologies of bounded growth, bounded independence and bounded arboricity have been very intensively studied in the last decade. (See, e.g.,  \cite{barenboim2010sublogarithmic,barenboim2013distributed,barenboim2017deterministic,lenzen2008can,schneider2008log}, and the references therein.) 
For many problems, the fastest algorithms were achieved in graphs of bounded growth, while graphs of bounded arboricity and bounded neighborhood independence have somewhat slower solutions, though are usually much better than in the general case. 
For example, MIS requires $\Omega(\sqrt{\log n/ \log \log n})$ rounds in line graphs \cite{kuhn2004cannot}, which are of neighborhood independence bounded by $2$, but can be solved in $O(\log^* n)$ rounds in graphs of bounded growth \cite{schneider2008log}.
In this work, however, we achieve the same result of $O(1)$ rounds for $O(1)$-backup-placement, not only for bounded growth, but also in the more general case  of bounded neighborhood independence.

Our results significantly improve upon previously known backup placement results. The best previous results are those for general graphs, due to Halldorsson et al. \cite{halldorsson2018distributed}. These are randomized algorithms with polylogarithmic running times and approximation ratio $O(\frac{\log n}{\log \log n})$. In the current paper, however, all our algorithms are {\em deterministic}. Our running time is $O(1)$ for graphs with constant neighborhood independence (i.e., dense graphs) and $O(\log n)$ for graphs with constant arboricity (i.e., sparse graphs). Our approximation ratio is $O(1)$ for graphs with constant neighborhood independence, as well as  for graphs with constant arboricity. Thus we significantly outperform the previous results, both in terms of running time and approximation ratio, for various dense and sparse graph families we mentioned above.

\begin{table}[]
\centering
\footnotesize
\begin{tabular}{|M{2.2cm}|M{2cm}|M{3cm}|M{3.5cm}|M{3.5cm}|}
\hline
\backslashbox[26mm]{Problem}{Topology} & Very dense: Graphs of bounded growth & Dense: Graphs of constant neighborhood independence & Sparse: Graphs of constant arboricity & General Graphs \\

\hline

\vspace{3mm} Maximal Independent Set \vspace{-3mm} & $\Theta(\log^{*} n)$  & $2^{O(\sqrt {\log n})}$ & $O(\frac{\log n}{\log \log n})$ & $2^{O(\sqrt {\log n})}$ \cite{panconesi1996complexity} \\
& \cite{linial1992locality,schneider2008log} & \cite{panconesi1996complexity} & \cite{barenboim2010sublogarithmic} &  $\Omega(\sqrt {\frac{\log n}{\log \log n}})$ \cite{kuhn2004cannot}  \\

\hline

\vspace{3mm} Maximal Matching \vspace{-3mm} & $\Theta(\log^* n)$ & $O(\log n)$ & $O(\frac{\log n}{\log \log n})$ \cite{barenboim2010sublogarithmic}  & $O(\log^3 n)$ \cite{fischer2017improved}  \\ 
& \cite{linial1992locality,schneider2008log} & 
\cite{barenboim2018distributed} & $\Omega(\sqrt {\frac{\log n}{\log \log n}})$ \cite{kuhn2004cannot} & $\Omega(\sqrt {\frac{\log n}{\log \log n}})$ \cite{kuhn2004cannot} \\ 

\hline

$O(\frac{\log n}{\log \log n})$-Backup Placement  & \multicolumn{4}{c|}{$O(\frac{\log^ 6 n}{\log^4 \log n}$) \ \ \ (randomized)  \ \ \cite{halldorsson2015bp} } \\ 

\hline

\vspace{3mm} $O(1)$-Backup Placement \vspace{-3mm} & $O(1)$ & $O(1)$  & $O(\log n)$  &  \\
& \textbf{This paper}  & \textbf{This paper}  & $\Omega(\sqrt {\frac{\log n}{\log \log n}})$  \ \ \ \ \ \ \ \   \ \ \ \ \ \ \ \ \ \ \ \ \textbf{This paper} & \\ 

\hline

\end{tabular}%
\caption{Comparison of running times of state-of-the-art symmetry-breaking and backup placement deterministic results. (With the exception of $O(\log n / \log \log n)$-backup-placement, for which the only non-trivial previous result is randomized.) \vspace{-0.5cm}}
\label{my-label}
\end{table}

Our results also extend the knowledge with respect to the question of "what can be computed locally?". This fundamental question, asking what can be computed in $O(1)$ rounds, was raised in the seminal paper of Naor and Stockmeyer \cite{naor1995can}. The importance of algorithms with a constant number of rounds follows from the fact that they are not affected by network size, nor by other parameters, such as maximum degree. Consequently, such algorithms are extremely scalable. Since the publication of Naor and Stokmeyer, several distributed algorithms with a constant number of rounds were devised, including weak-coloring \cite{kuhn2009weak,naor1995can}, $O(\Delta)$-forests-decomposition \cite{panconesi2001some}, minimum coloring approximation \cite{barenboim2012locality,barenboim2018fast}, certain network decompositions \cite{barenboim2012locality,barenboim2018fast}, approximate minimum dominating set \cite{lenzen2008can,barenboim2018fast}, and approximate minimum spanner \cite{barenboim2018fast}. However, several fundamental symmetry-breaking problems cannot be solved in $O(1)$-rounds, including Maximal Matching and Maximal Independent Set. A lower bound of $\Omega(\log^* n)$ holds for these problems on paths, cycles, graphs of bounded growth, graphs of bounded independence and more \cite{linial1992locality}. Thus, we find it intriguing that a related problem of $O(1)$-backup-placement can be solved in these and other graph topologies in a constant number of rounds. 

Our deterministic constant-time algorithm has an additional helpful property. Specifically, all the instructions are based on simple local rules. The combination of these three properties, namely, determinism, locality and simplicity, makes this algorithm very suitable for self-stabilization in faulty settings. We obtain a self-stabilizing variant that stabilizes just within $3$ rounds. 

\hspace{0.01cm}
\subsection{Our Upper-Bound Techniques}
Previous techniques for distributed backup placement and related problems were based on selfish improvement policies, generalizations of Maximal Matching, and Maximal Packing \cite{halldorsson2018distributed,halldorsson2015distributed}. In contrast, in this paper we employ an entirely different approach. Specifically, our main technical tool is a set of  trees or forests that cover all vertices (but not necessarily all edges) of the input graph. We call such a forest, which contains all vertices of the input graph $G$, a {\em forest cover}. In most of our constructions we employ a small number of forest covers, and so the number of edges is in the order of the number of vertices. Interestingly, such forest covers turn out to be very useful not only for graphs of small arboricity, but also for dense graphs in which the arboricity is very large, and the number of edges in the input graphs can be as large as $\Theta(n^2)$. Thus our approach is more general than the popular approach of \cite{barenboim2010sublogarithmic} that places all edges in forests, and works especially well on low-arboricity graphs. The new structures we devise in the current paper are used to efficiently coordinate backups, and reduce redundant load. For example, each parent in a tree can select one of its children for a backup. Since each vertex has at most one parent, this causes a load of only $1$ unit. But leafs have no children, and thus constitute a challenging instance. Selecting their parents for backups can cause excessive load. Thus, our algorithms select siblings for backups, and coordinate them in such a way that no sibling is selected by too many vertices. 

Another challenging instance is bipartite graphs. Even if one side of the partition has a constant degree, but on the other side the degree is unbounded, it is challenging to avoid the selection of the same vertex by many neighbors. We overcome this difficulty by gradually selecting vertices of sufficiently low degrees, and temporarily removing them and their neighbors from the graph. This produces additional vertices of small degrees, which allows us to repeat the process, until all vertices succeed to make a selection for a backup. Since in each stage vertices select neighbors of sufficiently low degrees, the resulting load is bounded. Since each time a constant fraction of vertices is removed, the algorithm terminates within $O(\log n)$ rounds. This idea can be generalized to work for graphs with constant arboricity. These graphs can be partitioned into few forests. All non-leaf vertices can be easily handled, by selecting their children for backups. Since each vertex has a bounded number of parents, the  load is bounded as well. It remains to handle the leafs. But the set of leafs and the set of non-leafs constitute a bipartite graph on which one side has a bounded degree. This can be handled as explained above. 

\hspace{0.01cm}
\subsection{Our Lower-Bound Techniques}
We consider two trees that are identical, except the very last layer of leafs. While one of the trees is $d$-regular, for a parameter $d > 0$, in the other one each leaf is a single child, rather than having $d - 2$ siblings, as in the $d$-regular tree. We prove that an optimal backup placement of the former tree is such that the root is selected for backups by all its children. On the other hand, for the latter tree, in the optimal backup placement the children of the root must select their own children, rather than the root. Since the trees are indistinguishable  (by the roots) within fewer rounds than the height of the trees, any deterministic algorithm with such a number of rounds will fail at least on one of them.

The above argument works for an optimal backup placement, but not for an approximate one. For the latter case, we provide a more sophisticated proof. We consider a graph in which all neighborhoods of sufficiently low radius look like one of the above-mentioned trees, possibly with single-child leafs in the last layer. Consequently, a backup placement algorithm executed on such a graph must assign a small number of backups to each vertex. (Each vertex can be seen as a root of a certain tree, in which most of its children do not select it.) But this turns out to be impossible, unless the trees are sufficiently high and the algorithm is invoked for sufficiently many rounds. 

\hspace{0.01cm}
\subsection{The Importance of Backup Placement in Wireless Networks}
Internet of Things (IoT) devices involve very differentiated requirements in terms of communication fashion, computation speed, memory limit, transmission rate and data storage capacity \cite{aggarwal2013internet}. A massive volume of heterogeneous data is created by those diverse devices periodically, and needs to be sent to other locations. As those devices may operate in a nonstop fashion and placed in geographically-diverse locations, they introduce a challenge in terms of communication overhead and storage capacity, since transfer and storage of data are resource-consuming activities within an IoT data management \cite{abu2013data}. Hence, these factors are crucial in data management techniques for IoT.

The challenging situation can be described as follows. As data is aggregated in a device or in the concentration storage points within the IoT, the amount of vacant storage becomes very limited. Thus the data usually needs to be sent further up the system, either to a storage relay node, or storage end-node \cite{fan2010scheme}. As wireless broadband communications is most likely to be used all along this process in the IoT world, there are also strict limitations to the amount of communication each power-limited device can perform. Moreover, since crashes of nodes, failures of communication links, and missing data are unavoidable in wireless networks, fault-tolerance becomes a key-issue. Among the causes of these constant failures are environmental factors, damaged communications links, data collision and overloaded nodes \cite{kakamanshadi2015survey}. Thus, the problems of efficient data transfer, storage and backup become major topics of interest in IoT
\cite{qin2016things}. 
A specific topic of interest is backup placement in wireless sensor networks, which was considered in \cite{oren2018distributed}.

\section{Backup Placement in Trees and Forest Covers of Dense Graphs}\label{intro2}
We begin with devising a procedure for computing $O(1)$-backup placement in trees. We assume that each vertex knows its parent in the tree, and each have unique ID of logarithmic length. Later we will show how to drop this assumption. The procedure receives a tree $T = (V,E)$ as input, and proceeds as follows. Each vertex that is not a leaf selects an arbitrary child for the backup placement. Each leaf selects its parent for the backup placement. All these selections are performed in parallel within a single round. This completes the description of the algorithm. Its pseudocode is provided in Algorithm \ref{algo1}. Its action is illustrated in Figure \ref{fig2} in Appendix A. The next lemma summarizes its correctness.

\begin{algorithm}[H]
\caption{Tree Backup Placement}
\label{algo1}
\begin{algorithmic}[1]
\Procedure{Tree-BP(Tree $T$)}{}
\State {\bf foreach} node $v \in T$ in parallel do:
\If{$v$ is not a leaf}
	\State $v.BP \gets Arbitrary(v.children)$.
\Else
	\State $v.BP \gets v.parent$
\EndIf
\EndProcedure
\end{algorithmic}
\end{algorithm}

\begin{lemma}
The algorithm computes an $O(1)$-backup placement of $T$.
\end{lemma}
\begin{proof}
Let $S$ be an optimal solution, and $S'$ is our solution for $T$. Each leaf must have selected its parent because this is its only neighbor. Let $f': V \rightarrow N$ denote the number of backups in each vertex $v \in V$ in our solution.  Let $f: V \rightarrow N$ denote the number of backups in each vertex $v \in V$ in the optimal solution. Next, we show that $f'(v) \leq f(v) + 1$, for each $v \in V$. It holds that $f(v)$ is at least the number of leaf children of $v$. On the other hand, $f'(v)$ is the number of  leaf children of $v$ plus at most $1$, since $v$'s non-leaf children do not select $v$, but $v$'s parent may select $v$. 
\end{proof}
The above algorithm can be directly extended to a forest of trees in which vertices know their parents, by executing the algorithm in parallel on all trees. More interestingly, we can extend this technique to graphs that are very dense, as opposed to trees and forests. To this end, we construct a forest cover of the input graph. (I.e., a subset of edges without cycles, such that each vertex in the input graph belongs to at least one edge in the subset.) The main idea of our algorithm is that each vertex whose ID is not a local maximum, selects an edge that connects with a neighbor of higher ID. Local maximum vertices select arbitrary neighbors. The detailed description of this algorithm and its analysis are relegated to Appendix \ref{appendix:forest}. We summarize this in Corollary \ref{cor:forestcover}. 

\begin{corollary} \label{cor:forestcover}
A forest cover of any input graph $G$ can be constructed within $O(1)$ rounds.
\end{corollary}


Next, we devise an algorithm that computes a backup placement for any graph $G$ in which a forest cover has been computed. For simplicity, we describe our algorithm for a tree $T$ in the forest. If there are more trees, the algorithm should be executed on all trees of the forest in parallel. (However, each vertex is aware only of the execution in the tree $T$ it belongs to.) 
We compute backup placement for $T$ as follows. Each non-leaf vertex in $T$ selects one of its children arbitrarily. Each leaf $v \in T$ selects a sibling $w$ of $v$ in $T$ with the property that $(v,w) \in E(G), ID(w) < ID(v)$, and there is no other sibling $z$ of $v$ such that $(v,z) \in E(G)$, and $ID(w) < ID(z) < ID(v)$. If there is no such sibling $w$, then $v$ selects its parent. This completes the description of the algorithm. Note that if there is such a sibling $w$, it is unique. The pseudocode is provided in Algorithm \ref{algo2}. Its action is illustrated in Figure \ref{fig3} in Appendix A. We analyze the algorithm below.

\begin{algorithm}[H]
\caption{General Backup Placement}
\label{algo2}
\begin{algorithmic}[1]
\Procedure{General-BP(Graph $G = (V,E)$, Tree $T$)}{}
\State {\bf foreach} node $v \in T$ in parallel do:
\If{$v$ is not a leaf}
	\State $v.BP \gets Arbitrary(v.children)$.
\ElsIf{$\exists w$ sibling of $v$, such that $(v,w) \in E$ and $(ID(w) < ID(v))$ $\land$ $\neg \exists z$ sibling of $v$, such that $(v,z) \in E$ and $(ID(w) < ID(z) < ID(v))$}
	\State $v.BP \gets w$.
\Else
	\State $v.BP \gets v.Parent$.
\EndIf

\EndProcedure
\end{algorithmic}
\end{algorithm}

First, we prove that for graphs with bounded neighborhood independence, the algorithm provides a good backup placement.

\begin{lemma}
Let $G$ be a graph with neighborhood independence at most $c$, for  $c > 0$, and $T$ be a tree in a forest cover of $G$. Then each vertex of $G$ is selected by at most $c$ of its \textbf{children} in $T$ in Algorithm \ref{algo2}. 
\end{lemma}
\begin{proof}
Assume for contradiction that $v \in V$ of $G$ has been selected by $c+1$ of its children in $T$. Let $V_{c+1}$ denote the set of those children. Since $G$ is a graph with neighborhood independence at most $c$, it implies that there must be at least two vertices $u, w \in V_{c+1}$ which are connected by an edge in G. Assume without loss of generality that $ID(u) < ID(w)$. In this case, according to our algorithm, $w$ must have selected either $u$ or some other sibling, but not its parent $v$. This is a contradiction.  
\end{proof}

\begin{lemma}
Let $G$ be a graph with neighborhood independence at most $c$, for $c > 0$, and $T$
be a tree in a forest cover of $G$. Then each vertex of $G$ is selected by at most $c$ of its \textbf{siblings} in $T$ in Algorithm \ref{algo2}.
\end{lemma}
\begin{proof}
If more than $c$ siblings select $v$, two of them are connected by an edge in $G$, since neighborhood independence is bounded by $c$. Denote these vertices by $u,w$. (Note that $(u,v) \in E(G), (u,w) \in E(G), (v,w) \in E(G)$). Assume without loss of generality that $ID(u) < ID(w)$. Since $u$ has selected $v$, it means that $ID(v) < ID(u)$. Therefore, $ID(v) < ID(u) < ID(w)$, and $w$ could not have selected $v$, since it should have selected a sibling with the closest $ID$ to its own. This is a contradiction.
\end{proof}


Since each vertex can be selected only by its children, by its siblings that are connected to it in $G$, and by its parent, and cannot be selected by vertices outside $T$, we obtain the next corollary. 

\begin{corollary}
Let $G$ be a graph with neighborhood independence at most $c > 0$, and $T$
 a tree in a for- est cover of $G$. Then each vertex of $T$ is selected by at most $2c + 1$ of its neighbors in $G$ in Algorithm \ref{algo2}.
\end{corollary}

The next lemma  analyzes the running time of the algorithm. The proof is relegated to Appendix \ref{o1}.

\begin{lemma}
The running time of the algorithm is $O(1)$.
\end{lemma}


By starting with a computation of a forest cover of the input graph $G$, and then invoking our algorithm on all trees of the forest cover of $G$ in parallel, we obtain the following result.

\begin{corollary}
For an input graph $G$ with neighborhood independence $c$, we compute backup placement with load at most $2c + 1$ in $O(1)$ rounds.
\end{corollary}
Next, we analyze the message complexity of our algorithm. The proof is in Appendix \ref{logn}.
\begin{lemma}
The algorithm can be implemented with messages of size $O(\log n)$ per link per round.
\end{lemma}


We note that for a variety of real-life network topologies, the neighborhood independence $c$ is bounded by a small constant. This includes Unit Disk Graphs that model wireless networks with the same transmission range for all nodes, and Bounded Disk Graphs in which transmission ranges may differ. We elaborate on this in Appendix \ref{appendix:udg-bdg}.
%
Our algorithm is applicable not only to wireless networks in which the number of independent nodes is bounded in each $r$-hop neighborhood, for a constant $r$ (as UDG and BDG), but also to more general networks. Specifically, it is sufficient that the number of neighbors is bounded only in the $1$-hop neighborhood. A notable example is line graphs, in which the number of independent nodes in $1$-hop neighborhoods is bounded by $2$, but it is unbounded in $r$-hop neighborhoods for $r \geq 2$. The proof of this is provided in Appendix \ref{appendix:udg-bdg}, Lemma \ref{lem:indln}.
Line graphs are an example of a family with {\em constant diversity} \cite{barenboim2017deterministic}. (Their diversity is bounded by $2$.) Graphs with diversity $c$ are those in which all vertices belong to at most $c$ cliques.  Hence, each vertex has at most $c$ independent neighbors. Thus, our algorithms are applicable also to this more general family of graphs with constant diversity. Moreover, we obtain a self-stabilizing variant of backup placement in graphs of bounded neighborhood independence that stabilizes just within $3$ rounds. We elaborate on that in Appendix \ref{selfstab}.
\vspace{-0.3cm}

\hspace{0.01cm}
\section{Backup Placement in Bipartite Graphs}
\label{BackupPlacementinBipartiteGraphs}

In this section we devise a backup placement algorithm for bipartite graphs $G = (U,V,E)$, in which the maximum degree of vertices in $U$ is bounded by a parameter $a$, and the maximum degree of vertices in $V$ is unbounded. Such graphs are motivated by client-server systems in which each client may connect to a bounded number of servers, but each server has an unrestricted number of clients. For now, we assume that all vertices know the load value $t$ of the optimum backup placement, i.e., the maximum number of selected vertices in an optimum solution. Vertices also know $a$ and $n$. Later, we relax these assumptions. The goal of our algorithm is obtaining a maximum load of $O(at)$. To this end, each vertex of $V$ may select for a backup an arbitrary neighbor in $U$. Since each vertex in $U$ has at most $a$ neighbors, the maximum load on vertices of $U$ is going to be at most $a$ as well. It remains to define an algorithm for the vertices of $U$. The algorithm is the following. It proceeds in phases. In each phase, each vertex in $U$ that has a neighbor in $V$ with degree at most $2at$, selects such a vertex. (The selection is arbitrary, if more than one such a neighbor exists.) Next, we remove from $V$ all vertices of degree at most $2at$, and also remove from $U$ all neighbors of vertices that have been removed from $V$. (Note that all these vertices of $U$ have already made selections in $V$. Note also that once some vertices are removed, certain vertex degrees become smaller.) Next, we proceed to the next phase, that is performed in the same way. Specifically, vertices in $U$ with neighbors of degree at most $2at$ select such neighbors, one neighbor each. Then, all vertices in $V$ of degree bounded by $2at$, and their neighbors, are removed. We repeat this until no vertices remain in $G$.
The pseudocode of the algorithm is provided in Algorithm \ref{algo3}. Its action is illustrated in Figure \ref{fig4} in Appendix A. Next we prove its correctness and analyze running time.

\begin{algorithm}[H]
\caption{The Bipartite Graph Distributed Backup Placement Algorithm}
\label{algo3}
\begin{algorithmic}[1]
\State for each $v \in V$, $v.allow\_backups = True$.
\Procedure{Bipartite-BP($G = (U,V,E), a, t$)}{}
\State In each round each vertex $v \in V$ in parallel does:
\If {$deg(v) < 2at$ and $deg(v) \neq 0$}
\State $v$ sends a message "allow backup" to all of its neighbors. 
\State Each vertex in $U$ that receives an "allow backup" message, selects for a backup an arbitrary neighbor that sent this message.
	\State $v.allow\_backups \gets False$
    \State $V = V \setminus v$, $U = U \setminus v.neighbors$
\EndIf
\EndProcedure
\end{algorithmic}

\end{algorithm}

\begin{lemma}
\label{lemmat}
In each phase, at most $1/2$ of remaining vertices in $V$ have more than $2at$ remaining neighbors each.
\end{lemma}

\begin{proof}
We denote the number of remaining vertices in $V$ in the beginning of each phase $i$ by $n_i$. We denote this set of remaining vertices of $V$ (respectively, $U$) by $V_i$ (resp., $U_i$). Denote also the number of remaining neighbors of a vertex $v$ in round $i$ by $deg_i(v)$. We prove our claim by contradiction. Assume that in a certain phase $i$ the number of vertices $v \in V_i$ with $deg_i(v) > 2at$ is more than $\frac{1}{2} n_i$. 
Thus, the degrees sum of vertices in $V_i$ in this stage is at least $\sum deg_i(v), v \in V_i > \frac{n_i}{2} \cdot 2at = n_{i} \cdot at$. Recall that $t$ is a parameter that quantifies the optimal solution. It is important to emphasize that the meaning of $t$ as the optimal backup is the following. For each $u \in U$, and in particular $u \in U_i$, there exists a selection of some $v \in V$ for backup, such that after those selections have been made, the degree of each $v \in V$ in the subgraph induced by the selected edges is $deg(v) \leq t$. Note also that all neighbors of vertices of $U_i$ are in $V_i$. Indeed, any vertex in $U$ with a neighbor outside of $V_i$ has been removed in an earlier stage. Hence, in the optimal solution, all selection of vertices of $U_i$ are vertices of $V_i$, and the maximum load is $t$. Thus, the size $|U_i|$ of $U_i$, in the optimal solution, is
 $|U_i| \leq n_{i} \cdot t$. \ \ \ \ \ \ \ \ \ \ \ \ \ \ \ \ \ \ \ \ \ \ \ \ \ \ \ \ \ \ \ \ \ \ \ \ \ \ \ \ \ \ \ \ $(*)$ \\ But each vertex in $U_i$ has degree at most $a$, and $\sum deg_i(v), v \in V_i$ is greater than $n_{i} \cdot at$, by the above assumption. Since in phase $i$ all remaining neighbors of vertices in $V_i$ are in $U_i$, 
it follows that the number of vertices in $U_i$ is greater than $n_{i} \cdot at / a = n_{i} \cdot t$. This is a contradiction to $(*)$.
\end{proof}

\begin{lemma}
The algorithm terminates within $O(\log n)$ rounds.
\end{lemma}

\begin{proof}
As proven in the lemma above, at any phase $i$ the number of vertices $v \in V$ with $deg_i(v) > 2at$ is at most $\frac{1}{2}$ of remaining vertices in that phase. Thus, in the first round of the algorithm, the size of the subset $\{ v \in V \mid deg_1(v) > 2at \}$ is at most $\frac{1}{2} n_1 = \frac{1}{2}|V|$. Similarly, in the next round, the size of the subset $\{ v \in V \mid deg_2(v) > 2at \}$ is at most $\frac{1}{2} n_2 = \frac{1}{4} n_1$, and so forth. After $R$ rounds, for a positive integer $R$, the number of vertices in $V$ with degree greater than $2at$ is bounded by $n_1 / 2^R$. Since the algorithm completes all backup placements in the round when no vertices in $V$ with degree greater than $2at$ remain, it terminates within $O(\log n_1) = O(\log n)$ rounds.
\end{proof}

\begin{lemma}
Each vertex in $V$ is selected by at most $2at$ vertices in $U$.
\end{lemma}

\begin{proof}
In each round vertices in $U$ select neighbors in $V$ only if these neighbors have degree at most $2at$. Moreover, each vertex in $V$ is selected in the same round by all neighbors that choose it during the algorithm. (This is because in the end of the round when a vertex is selected for the first time, it is removed from $V$.) But a vertex cannot be selected by more than $2at$ neighbors of $U$, since in that round the vertex degree is bounded by $2at$.
\end{proof}

Next, we extend our algorithm to the scenario that the optimum value $t$ is not known to the vertices. (But vertices know $a$. In Section \ref{sc:arb} we elaborate on the reason of this assumption regarding knowledge of $a$.) In this case when $t$ is not known, we invoke our algorithm several times, starting with an estimation $t' = 1$ for $t$, and doubling $t'$ after each invocation. As long as $t' < t$, an invocation may remove only some of the vertices, in contrast to the case $t' \geq t$, in which all vertices are removed, as shown in the proof of Lemma \ref{lemmat}. Thus, each invocation is performed on the set of remaining vertices from the previous invocation. In the first time when $t'$ becomes at least $t$, all remaining vertices are removed, and we are done. Note that in this stage it holds that $t' < 2t$. In each invocation a vertex is selected by at most $2at'$ neighbors, and so the maximum number of selections of the same vertex in all invocations is $O(at)$. The number of rounds required for this computation is $O(\log n \log t)$. We summarize this in the next theorem.

\begin{theorem}
Suppose that the optimum backup placement value $t$ is not known to the vertices before execution. Then we compute a backup placement with a load of $O(at)$ in $O(\log n \log t) = O(\log^2 n)$ rounds.
\end{theorem}

We note that the $O(\log t)$ executions can be performed in parallel. To this end, we assign arrays $A_u$ of size $\log n$ to each vertex $u \in U$, and execute the algorithm independently for $\log n$ times in parallel. (We do not know t, but we know that $t < n$.) Each such execution requires $O(\log n)$ time, but their parallel invocations result in $O(\log n)$ rounds overall, rather than $O(\log^2 n)$. Now, we have $\log n$ results for each vertex in $U$. We note that $A_u[i]$ may have not reached an answer, if $2^i < t$. Otherwise, $A_u[i]$ must contain a selection for $u$, by the correctness of the original algorithm, that works with any value that is at least the optimum $t$. For each vertex $u \in U$, we take the result for the smallest $i$ for which a backup placement $A_u[i]$ has been found. We know that $2^i < 2t$, where $t$ is the optimum.  
Next, we analyze how many vertices in $U$ may select the same vertex in $V$ in the worst case. Consider such a vertex $v \in V$, and denote $u_1,u_2,...,u_q$ the set of vertices that selected $v$. Our goal is to bound $q$. Recall that each selection of a vertex $u_j, j \in [q]$ is stored in $A_{u_j}[i]$, where $i < \log (2t)$. Moreover, there are at most $2a$ vertices that selected $v$ in $\{A_{u_1}[1], A_{u_2}[1],...,A_{u_q}[1]\}$, at most $4a$ vertices that selected $v$ in $\{A_{u_1}[2], A_{u_2}[2],...,A_{u_q}[2]\}$, at most $8a$ vertices that selected $v$ in $\{A_{u_1}[3], A_{u_2}[3],...,A_{u_q}[3]\}$, etc. In general, for $i > 0$, there at most $2^i \cdot a$ vertices that selected $v$ in $\{A_{u_1}[i], A_{u_2}[i],...,A_{u_q}[i]\}$. Thus, the overall number of vertices that selected $v$ in all these arrays cells with indices $1,2,...,i$ is bounded by $2a \cdot \sum_{j = 1}^{i} 2^j = 2a \cdot (2^{i + 1} - 2) < 8 a \cdot t$. 

\begin{corollary} \label{cor:selections}
The number of selections of the same vertex in $V$ by vertices in $U$ is bounded by $8at$.  
\end{corollary}

Finally, we note that since the dependency on $n$ is logarithmic, vertices need to know only a polynomial estimation of $n$, rather than its exact value. We summarize the results of this section in the next theorem.
\begin{theorem} \label{thm:bipartiteAlg}
Given a bipartite graph $G = (U,V,E)$  with maximum degree $a$ in $U$, one can compute a backup placement with a load of $O(at)$ in $O(\log n)$ rounds, even if the optimum load $t$ is not known to the vertices, and vertices have only polynomial estimation of $n$.
\end{theorem}

\section{Backup Placement in Graphs of Bounded Arboricity}
\label{sc:arb}
Following our method for backup placement in bipartite graphs, we can also prove $O(\log n)$ time complexity for backup placement in graphs of constant arboricity, and in particular, in planar graphs. We still assume that all vertices know the load value $t$ of the optimum backup placement, i.e., the maximum number of selected vertices in an optimum solution. The goal of our algorithm is obtaining a maximum load of $O(a \cdot t)$, where $a$ is the arboricity of the input graph. We note that our algorithm can be extended to the scenario where the optimum load $t$ is not known, similarly to Section \ref{BackupPlacementinBipartiteGraphs}. Nevertheless, vertices still need to know $a$. We note, however, that this information is often available to the vertices. For example, in planar graphs, vertices know that the arboricity $a$ is bounded by $3$ \cite{nash1961edge,nash1964decomposition}.

In order to fulfill our goal, we use the Procedure Partition algorithm \cite{barenboim2010sublogarithmic,barenboim2013distributed}, which receives as input a graph $G$ with arboricity $a$, and partitions it into $\ell = O(\log n)$ sets, $H_1,H_2,...,H_{\ell}$, such that the number of neighbors of a vertex $v \in H_i$, $i \in [\ell]$, in the set $H_i \cup H_{i + 1} \cup ... \cup H_{\ell}$ is at most $(2 + \epsilon)a$, for an arbitrarily small constant $\epsilon > 0$. The time complexity of this algorithm is $O(\log n)$.
For the sake of simplicity, we set $\epsilon=1$. Consequently, the algorithm partitions the vertices of $G$ into $\ceil{\frac{2}{\epsilon}\log n} = 2\log n$ sets, such that each vertex $v$ has at most $3a$ neighbors in the union of sets with greater or equal index to the $H$-set of $v$. 

We note that the sets $H_1,H_2,...,H_{\ell}$ obtained using Procedure Partition have the following property. For each vertex in $H_i$ with $i > 1$, there must be a neighbor in a set with a smaller $H$-index. For example, for a vertex in $H_2$ there must be a neighbor in $H_1$, and for a vertex in $H_3$ there must be a neighbor in $H_2$ or $H_1$, and so forth. The reason, according to Procedure Partition \cite{barenboim2010sublogarithmic}, is that all vertices that did not join $H_1$ have an initial degree that is greater than $3a$, and once they join an $H$-set of a greater index, their degree becomes at most $3a$. This means that each of them has a neighbor that has been removed earlier, i.e., moved to an $H$-set with a smaller index. For the backup placement problem, this implies that it is possible to create a backup placement from any $H_i, i>1$, to a vertex in $H_j, j<i$. To this end, all vertices in $H_i, i > 1$, select in parallel neighbors in sets with smaller indices than their own, one neighbor each. Since the number of neighbors in sets of greater or equal index is bounded by $3a$, each vertex is selected by at most $3a$ neighbors in this stage.
In the case of planar graphs, $a \leq 3$ \cite{nash1961edge,nash1964decomposition}. 
Thus, in planar graphs all vertices in $H_2,...,H_{\ell}$ make selections with a load of up to $3a = 9$ on any vertex in $G$.

It remains to select placements for vertices in $H_1$. For each vertex in $H_1$ that has a neighbor in $H_1$, we select such a neighbor arbitrarily. Since the maximum degree in $H_1$ is bounded by $3a$, each vertex can be selected by at most $3a$ neighbors. However, some vertices in $H_1$ may have no neighbors in this set, but only neighbors in sets of greater indices. Next, we describe a solution for these vertices, and thus, complete the description of placements selection for all vertices in $H_1,H_2,...,H_{\ell}$.
We denote $U = \{ \mbox{all vertices in } H_1 \mbox{ that have no neighbors in } H_1 \}$, $V = H_2 \cup H_3 \cup ... \cup H_{\ell}$, $E' = \{(u,v), u \in U, v \in V, (u,v) \in E\}$, and execute our bipartite algorithm from Section \ref{BackupPlacementinBipartiteGraphs} on $G' = (U,V,E')$. Consequently, each vertex in $U$ finds a backup placement in $V = H_2 \cup H_3 \cup ... \cup H_{\ell}$. In addition, according to Procedure Partition, each vertex $v \in V$ has a neighbor in a smaller-index $H$-set, and thus each $v \in V$ can perform backup placement to such a neighbor in a smaller-index set. Finally, the vertices that are neither in $U$ nor in $V$ belong to $H_1$, and have neighbors in $H_1$. They select arbitrary neighbors in $H_1$ for backup placements. This completes the description of placements selection for all vertices in the input graph. 

The pseudocode is provided in Algorithm \ref{algo4}.

\begin{algorithm}[H]

\caption{The Bounded Arboricity Graph Distributed Backup Placement Algorithm}
\label{algo4}
\begin{algorithmic}[1]
\Procedure{Bounded-Arboricity-BP($G = (V,E), a, t$)}{}
\State for each $v \in V$, $v.allow\_backups = True$.
\State $H_1,H_2,...,H_{\ell}$ = Procedure-Partition(G,a)
\State each vertex $v \in H_i, \ell \geq i>1$ in parallel does:
\State \qquad select one arbitrary neighbor in $H_j, j<i$
\State each vertex $v \in H_1$ in parallel does:
\If {$v$ has a neighbor in $H_1$}
\State $v$ selects a neighbor in $H_1$ arbitrarily
\EndIf
\State BIPARTITE-BP($U = \{ \mbox{all vertices in } H_1 \mbox{ that have no neighbors in } H_1 \}$, $V = H_2 \cup H_3 \cup ... \cup H_{\ell}$, $E' = \{(u,v), u \in U, v \in V, (u,v) \in E\}$, \ $a$,   \  $t$)
\EndProcedure
\end{algorithmic}
\end{algorithm}

The running time of the algorithm is dominated by the execution of Procedure Partition, followed by our Bipartite algorithm of Section \ref{BackupPlacementinBipartiteGraphs}. Both these algorithms require $O(\log n)$ time. It remains to analyze the load of the algorithm. This is done in the next theorem.

\begin{theorem}
For any positive parameter $a$, we compute Backup Placement in graphs of arboricity $a$  within $O(\log n)$ time. The resulting load is $O(at)$, where $t$ is the optimum load of the input graph.
\end{theorem}
\begin{proof}
The running time follows from Theorem \ref{thm:bipartiteAlg}, and from the running time $O(\log n)$ of Procedure Partition of \cite{barenboim2010sublogarithmic}. Next, we analyze the load.
 The maximum number of backups assigned to a vertex in each step, except the step of executing the Bipartite algorithm, is $3a$. The maximum number of backups of the Bipartite algorithm is $24ta$, where $t$ is the load of the optimal backup placement of $G'$. (See Corollary \ref{cor:selections}. Indeed, each vertex in $U$ has at most $3a$ neighbors in $V$.) Our goal, however, is to analyze the optimal load of $G$, rather than $G'$. Nevertheless, we next argue that the optimal load of a backup placement for $G$ cannot be significantly smaller than that for $G'$. Thus, the above solution with load at most $3a + 24ta \leq 27ta$ is an $O(a)$-approximation to the optimal backup placement of $G$ as well.
 
Observe that the optimal load on vertices of $V$ in $G'$ is not greater than the load on these vertices in the optimal solution for $G$. Otherwise, the selections of vertices of $U$ in $G$ constitute a solution for $G'$ with a smaller load on $V$ than the optimal one. This is a contradiction. Thus, the execution of our Bipartite algorithm on $(U,V,E')$ results in an $O(a)$-approximation for the optimal load on $V$ in $G$. In addition to the selections made during this execution, our algorithm performs additional selections, made by vertices that do not belong to $U$. However, all these additional selections cause an additional load of at most $3a$ in each vertex. Hence, we achieve
 an $O(a)$-approximation to the optimal solution for $G$.  
 
\end{proof}

\section{Lower Bounds}
We demonstrate that our upper bounds are not far from optimal, by proving lower bounds for exact and approximate backup placement. We begin with a lower bound for an exact (optimal load) solution. Our proof deals with a slightly modified version of the problem, in which in addition to computing backup placement, each vertex outputs the number of neighbors that selected it. Note that a solution to the original problem can be transformed to a solution for the modified one, within an additional single round. Specifically, given a backup placement solution, all vertices communicate in parallel with the vertices they have selected. The number of communicating neighbors with each vertex is the required solution. Thus, a lower bound of $R$ rounds for the modified version implies a lower bound of $R-1$ for the original one.
\begin{lemma}
Optimal backup placement in unoriented trees of height $h$ requires $\Omega(h)$ rounds, for deterministic algorithms. In terms of the number of vertices $n$, the running time is $\Omega(\log n)$.
\end{lemma}
\begin{proof}
In order to prove our claim we employ two distinct trees. (1) A fully balanced tree, in which each vertex (except the leafs) has a constant number of children $d$, and the height $h$ is odd. (This makes the number of levels in the tree even.) (2) A "single leafs" tree with height $h$, in which each non-leaf vertex (except the root) has $d-1$ siblings, but all leafs are single children, e.g., they have no siblings at all. The only difference between the fully balanced tree and the single leafs tree is in the $h$ level, the furthest level from the root. (See Figure \ref{fig7} in Appendix A).
The optimal solution for each of the above trees is: \\ Tree (1): all vertices in even layers perform backups to their parents, while all vertices in odd layers perform backups to their children. In particular, all leafs perform backups to their parents. \\  Tree (2): each vertex which is not a leaf selects one of its children for backup, and only the leafs select their parents for backups.
Consequently,  in the optimal backup placement in Tree (1) the load is $d$, while in Tree (2) the load is $1$ in all levels except $h-1$, and the load is $2$ in level $h-1$. The reason for the load in level $h-1$ of Tree (2) is that both the single leaf and its grand-parent send their backups to the parent of the leaf. Next, we consider the root vertices in the two trees. Notice that in Tree (1) all the backups in level $1$ must be sent to the root, and thus its load is $d$. (Otherwise, some vertex is going to have a load greater than $d$.) On the other hand, in Tree (2) there are no backups at all of the root's children at the root, since they select their own children. Thus  the load on the root is $0$.

It follows that the optimal solution in both of the trees depends solely on the $h$ level. In order to prove that the lower bound for optimal solution of the backup placement problem in this case is $\Omega(h)$, we prove by contradiction that it is impossible to achieve the optimal solution while ignoring the $h$ level. Assume that there is a deterministic optimal solution with time complexity smaller than $h - 1$. When this algorithm is executed on the two trees, the roots of both trees perform exactly the same computations, since the subtrees of height $h - 1$ are the same in both trees.
Therefore, the roots compute the same solutions in both trees (1) and (2). However, the correct solutions of the roots are distinct in each of these trees. Specifically, the output of root (1) is "load $d$", while the output of root (2) is "load $0$".  Hence, this execution of the algorithm is correct only in at most one of the trees, while wrong in at least one of them necessarily. Thus, there is no optimal solution with time complexity lower than $h - 1$. For $d$-regular trees with a constant $d$, the height is $\Omega(\log n)$.
\end{proof}


\noindent \ \ Next, we provide a lower bound for approximate backup placement.

\begin{theorem}
$O(1)$-backup placement requires $\Omega(\sqrt {\frac{\log n}{\log \log n}})$ rounds, for deterministic algorithms.
\end{theorem}
\begin{proof}
The proof is by contradiction.
In order to prove our lower bound we invoke any distributed algorithm $A$ for $O(1)$-backup placement with running time $k = o(\sqrt{ \log n / \log \log n})$  on the Kuhn-Moscibroda-Wattenhofer \cite{kuhn2004cannot} graph $G$. This graph has various helpful properties, but we are interested in two specific properties that are crucial  to our proof.  (1) The girth of the graph is at least 2k + 2, and (2) the minimum degree in the graph is at least $2$. The first property means that the view of a vertex $v$ with respect to radius $k$ (i.e., the subgraph of $G$ induced by all vertices at distance at most $k$ from $v$) is a tree. The second property implies that in such a view, all leafs are at distance $k$ from $v$, but not closer to $v$. This is because the degrees are preserved in the tree, except the very last layer at distance $k$ from $v$, in which the degrees become equal to $1$. We denote this tree by $T$. The vertex $v$ is considered to be the root of $T$.

Next, consider a tree $T'$ obtained by adding a single child to each vertex at distance $k$ from the root in $T$. These new vertices constitute the set of leafs of $T'$, and they all are at distance $k + 1$ from the root of $T'$. Within distance $k$ from the root the trees $T$ and $T'$ are indistinguishable. Since the optimal backup placement on $T'$ has load $2$, the Algorithm $A$ invoked on either $T'$, $T$ or $G$ must assign at most $O(1)$ backups to $v$. This is because $T'$, $T$ and $G$ are indistinguishable to $A$ that is executed for $k =  o(\sqrt{ \log n / \log \log n})$ rounds. Observe that this is true for all vertices $v$ in $G$, once $A$ terminates on $G$. In other words, $A$ computes an $O(1)$-backup placement of $G$.

Now consider the subgraph $G' = (V,E')$ of $G$, where $E'$ is the set of all edges $(u,w)$, such that either $u$ selected $w$ for a backup, $w$ selected $u$, or both. Since each vertex selects just a single neighbor for a backup, and each vertex is selected by $O(1)$ neighbors, the maximum degree of $G'$ is $O(1)$. However, it contains all vertices of $G$. Thus $G'$ is a constant approximation to a minimum vertex cover of $G$. (Take a minimum vertex cover of $G$ and add all neighbors in $G'$ of its vertices. The result is $V(G') = V(G)$, and the size is greater by a multiplicative factor of $O(1)$.) But computing such an approximation requires $\Omega(\sqrt{\frac{\log n}{\log \log n}})$ time on $G$ \cite{kuhn2004cannot}.

Note: it may be possible that such a graph $G'$ that covers $G$ and has maximum degree $O(1)$ does not exist at all. In this case the assumption of existence of an algorithm $A$ with running time $o(\sqrt{ \log n / \log \log n})$  leads to a contradiction (the existence of such a graph). 
\end{proof}
\noindent \ \ Note that the lower bound of \cite{kuhn2004cannot} that we employed
applies not only to a constant-factor approximation, but also to a polylogarithmic one. Consequently, our arguments are directly extended to the next result.
\begin{corollary}
Backup placement with polylogarithmic approximation factor requires $\Omega(\sqrt {\frac{\log n}{\log \log n}})$ rounds, for any deterministic algorithm.
\end{corollary}


\bibliographystyle{plain} 
\bibliography{bibliography} 

\newpage

\begin{appendices}

\section{Figures}

\vspace{-1cm}


\begin{figure}[H]
\centering
\begin{tikzpicture}

  \tikzstyle{vertex}=[circle,fill=black!25,minimum size=12pt,inner sep=2pt]
  
  \node[vertex] (G_1) at (1,0) {1};
  \node[vertex] (G_2) at (-1,0) {2};
  \node[vertex] (G_3) at (1,-2) {3};
  \node[vertex] (G_4) at (-1,-2) {4};
  \node[vertex] (G_5) at (-2.5,-1) {5};
  \node[vertex] (G_6) at (2.5,-1) {6};
  
  \draw [line width=0.2mm, black] (G_1) -- (G_2); 
  \draw [line width=0.2mm, black] (G_2) -- (G_5);
  \draw [line width=0.2mm, black] (G_5) -- (G_4);
  \draw [line width=0.2mm, black] (G_4) -- (G_3);
  \draw [line width=0.2mm, black] (G_3) -- (G_6);
  \draw [line width=0.2mm, black] (G_6) -- (G_1);

  \draw [->, line width=0.2mm, blue] [dashed] (G_1) to[out=135,in=45] (G_2);
  \draw [->, line width=0.2mm, blue] [dashed] (G_2) to[out=180,in=90] (G_5);
  \draw [->, line width=0.2mm, blue] [dashed] (G_5) to[out=270,in=180] (G_4);
  \draw [->, line width=0.2mm, blue] [dashed] (G_4) to[out=315,in=225] (G_3);
  \draw [->, line width=0.2mm, blue] [dashed] (G_3) to[out=0,in=270] (G_6);
  \draw [->, line width=0.2mm, blue] [dashed] (G_6) to[out=90,in=0] (G_1);

\end{tikzpicture}
\hspace{1cm}
\begin{tikzpicture}

  \tikzstyle{vertex}=[circle,fill=black!25,minimum size=12pt,inner sep=2pt]
  
  \node[vertex] (G_1) at (0,0) {1};
  
  \node[vertex] (G_2) at (-2,-1) {2};
  \node[vertex] (G_3) at (2,-1) {3};
  
  \node[vertex] (G_6) at (-1,-2) {6};
  \node[vertex] (G_7) at (1,-2) {7};
  \node[vertex] (G_61) at (-2,-2) {5};
  \node[vertex] (G_71) at (2,-2) {8};
  \node[vertex] (G_62) at (-3,-2) {4};
  \node[vertex] (G_72) at (3,-2) {9};

  \draw [-, line width=0.3mm, black] (G_1) -- (G_2); 
  \draw [-, line width=0.3mm, black] (G_1) -- (G_3);
  
  \draw [-, line width=0.3mm, black] (G_2) -- (G_6);
  \draw [-, line width=0.3mm, black] (G_3) -- (G_7);
  \draw [-, line width=0.3mm, black] (G_2) -- (G_61);
  \draw [-, line width=0.3mm, black] (G_3) -- (G_71);
  \draw [-, line width=0.3mm, black] (G_2) -- (G_62);
  \draw [-, line width=0.3mm, black] (G_3) -- (G_72);
  
\draw [->, line width=0.3mm, blue] [dashed] (G_1) to[out=270,in=0] (G_2);

\draw [->, line width=0.3mm, blue] [dashed] (G_2) to[out=90,in=180] (G_1);
\draw [->, line width=0.3mm, blue] [dashed] (G_3) to[out=90,in=0] (G_1);

\draw [->, line width=0.3mm, blue] [dashed] (G_6) to[out=0,in=0] (G_2);
\draw [->, line width=0.3mm, blue] [dashed] (G_62) to[out=180,in=180] (G_2);

\draw [->, line width=0.3mm, blue] [dashed] (G_72) to[out=0,in=0] (G_3);
\draw [->, line width=0.3mm, blue] [dashed] (G_7) to[out=180,in=180] (G_3);

\draw [->, line width=0.3mm, blue] [dashed] (G_71) to[out=0,in=280] (G_3);
\draw [->, line width=0.3mm, blue] [dashed] (G_61) to[out=180,in=260] (G_2);

\end{tikzpicture}
\vspace{-0.3cm}
\caption{Optimal backup placement (in blue)  in $C_6$ cycle graph (left) and a rooted tree (right)} 
\label{fig1}
\end{figure}

\vspace{-1.5cm}

\begin{figure}[H]
    \centering
    \begin{minipage}{0.45\textwidth}
        \centering

\begin{tikzpicture}
  \tikzstyle{vertex}=[circle,fill=black!25,minimum size=12pt,inner sep=2pt]
  \node[vertex] (G_1) at (0,0) {0};
  \node[vertex] (G_2) at (-2,-2) {1};
  \node         (G_3) at (0,-2) {\ldots};
  \node[vertex] (G_4) at (2,-2) {n};
  \draw [->, line width=0.3mm, black] (G_1) -- (G_2); 
  \draw [->, line width=0.3mm, blue] [dashed] (G_1) -- (G_3);
  \draw [->, line width=0.3mm, black] (G_1) -- (G_4);
  \draw [->, line width=0.3mm, blue] [dashed] (G_2) to[out=180,in=90] (G_1); 
  \draw [->, line width=0.3mm, blue] [dashed] (G_4) to[out=0,in=90] (G_1);
  \draw [->, line width=0.3mm, blue] [dashed] (G_3) to[out=0,in=0] (G_1);
  \draw [->, line width=0.3mm, blue] [dashed] (G_3) to[out=180,in=180] (G_1);
  
\end{tikzpicture}
\vspace{-0.3cm}
\caption{Tree Backup Placement}
\label{fig2}
\end{minipage}\hfill
    \begin{minipage}{0.5\textwidth}
        \centering
\begin{tikzpicture}
  \tikzstyle{vertex}=[circle,fill=black!25,minimum size=12pt,inner sep=2pt]
  \node[vertex] (G_1) at (0,0) {0};
  \node[vertex] (G_2) at (-2,-2) {1};
  \node         (G_3) at (0,-2) {\ldots};
  \node[vertex] (G_4) at (2,-2) {n};
  \draw [->, line width=0.3mm, black] (G_1) -- (G_2); 
  \draw [->, line width=0.3mm, blue] [dashed] (G_1) -- (G_3);
  \draw [->, line width=0.3mm, black] (G_1) -- (G_4);
  \draw [->, line width=0.3mm, blue] [dashed] (G_2) to[out=180,in=90] (G_1); 
  \draw [->, line width=0.3mm, blue] [dashed] (G_3) -- (G_2);
  \draw [->, line width=0.3mm, blue] [dashed] (G_4) -- (G_3);
\end{tikzpicture}
\vspace{-0.3cm}
        \caption{General Backup Placement}
        \label{fig3}
    \end{minipage}
\end{figure}

\vspace{-0.6cm}

\begin{figure}[htb]
    \centering
\begin{minipage}[b]{.3\linewidth}
\centering
        \begin{tikzpicture}[
    node distance = 7mm and 21mm,
      start chain = going below,
         V/.style = {circle, draw, 
                     fill=#1, 
                     inner sep=0pt, minimum size=3mm,
                     node contents={}},
 every fit/.style = {ellipse, draw=#1, inner ysep=-1mm, 
                     inner xsep=5mm},
                    ]
\foreach \i in {4,...,0} 
{
    \node (n1\i) [V=myblue,on chain,
                  label={[text=myblue]left:$u_{\i}$}];
}

\foreach \i in {2,...,0} 
{
    \node (n2\i) [V=mygreen, above right=3.5mm and 22mm of n1\i,
                  label={[text=mygreen]right:$v_{\i}$}];
}
\node [myblue,fit=(n14) (n10),label=above:$U$] {};
\node [mygreen,fit=(n22) (n20),label=above:$V$] {};
\node [myblue,fit=(n14) (n10),label=below:$m$] {};
\node [mygreen,fit=(n22) (n20),label=below:$n$] {};


\draw [->, line width=0.3mm, black] [] (n10) to (n20);
\draw [->, line width=0.3mm, black] [] (n10) to (n21);

\draw [->, line width=0.3mm, black] [] (n11) to (n20);
\draw [->, line width=0.3mm, black] [] (n11) to (n21);

\draw [->, line width=0.3mm, black] [] (n12) to (n20);
\draw [->, line width=0.3mm, black] [] (n12) to (n21);

\draw [->, line width=0.3mm, black] [] (n13) to (n20);

\draw [->, line width=0.3mm, black] [] (n14) to (n20);
\draw [->, line width=0.3mm, black] [] (n14) to (n22);

    \end{tikzpicture}
    \captionof{figure}[One figure]{Round 1.}
\end{minipage}
\hfill
\begin{minipage}[b]{.3\linewidth}
\centering
        \begin{tikzpicture}[
    node distance = 7mm and 21mm,
      start chain = going below,
         V/.style = {circle, draw, 
                     fill=#1, 
                     inner sep=0pt, minimum size=3mm,
                     node contents={}},
 every fit/.style = {ellipse, draw=#1, inner ysep=-1mm, 
                     inner xsep=5mm},
                    ]
\foreach \i in {4,...,0} 
{
    \node (n1\i) [V=myblue,on chain,
                  label={[text=myblue]left:$u_{\i}$}];
}
                  
\foreach \i in {2,...,0} 
{
    \node (n2\i) [V=mygreen, above right=3.5mm and 22mm of n1\i,
                  label={[text=mygreen]right:$v_{\i}$}];
}

                  
    \node (n21) [V=myempty, above right=3.5mm and 22mm of n11,
                  label={[text=mygreen]right:$v_{1}$}];
    \node (n22) [V=myempty, above right=3.5mm and 22mm of n12,
                  label={[text=mygreen]right:$v_{2}$}];

\node [myblue,fit=(n14) (n10),label=above:$U$] {};
\node [mygreen,fit=(n22) (n20),label=above:$V$] {};
\node [myblue,fit=(n14) (n10),label=below:$m$] {};
\node [mygreen,fit=(n22) (n20),label=below:$n$] {};

\draw [->, line width=0.3mm, blue] [dashed] (n14) to (n22);

\draw [->, line width=0.3mm, black] [] (n13) to (n20);


\draw [->, line width=0.3mm, blue] [dashed] (n10) to (n21);

\draw [->, line width=0.3mm, blue] [dashed] (n11) to (n21);

\draw [->, line width=0.3mm, blue] [dashed] (n12) to (n21);

    \end{tikzpicture}
  \captionof{figure}[One figure]{Round 2.}
\end{minipage}
\hfill
\begin{minipage}[b]{.3\linewidth}
\centering
        \begin{tikzpicture}[
    node distance = 7mm and 21mm,
      start chain = going below,
         V/.style = {circle, draw, 
                     fill=#1, 
                     inner sep=0pt, minimum size=3mm,
                     node contents={}},
 every fit/.style = {ellipse, draw=#1, inner ysep=-1mm, 
                     inner xsep=5mm},
                    ]
\foreach \i in {4,...,0} 
{
    \node (n1\i) [V=myblue,on chain,
                  label={[text=myblue]left:$u_{\i}$}];
}

\foreach \i in {2,...,0} 
{
    \node (n2\i) [V=myempty, above right=3.5mm and 22mm of n1\i,
                  label={[text=mygreen]right:$v_{\i}$}];
}


\node [myblue,fit=(n14) (n10),label=above:$U$] {};
\node [mygreen,fit=(n22) (n20),label=above:$V$] {};
\node [myblue,fit=(n14) (n10),label=below:$m$] {};
\node [mygreen,fit=(n22) (n20),label=below:$n$] {};

\draw [->, line width=0.3mm, blue] [dashed] (n14) to (n22);

\draw [->, line width=0.3mm, blue] [dashed] (n13) to (n20);


\draw [->, line width=0.3mm, blue] [dashed] (n10) to (n21);

\draw [->, line width=0.3mm, blue] [dashed] (n11) to (n21);

\draw [->, line width=0.3mm, blue] [dashed] (n12) to (n21);

    \end{tikzpicture}
  \captionof{figure}[One figure]{Round 3.}
\end{minipage}
\vspace{-0.2cm}
\caption{The Bipartite Graph Distributed Backup Placement Execution: An example of a bipartite graph, with $a = 2$ and an estimation $t' = 1$ for $t$, and we allow up to $2 \cdot a t'=4$ backups per $v \in V$. Black arrows represent edges of the input graph, while blue ones represent selections of backup placement.\vspace{1cm}}
\label{fig4}
\end{figure}

\vspace{-1.5cm}

\begin{figure}[H]
\centering
\begin{minipage}[b]{.5\linewidth}
\begin{tikzpicture}
  \tikzstyle{vertex}=[circle,fill=black!25,minimum size=12pt,inner sep=2pt]
  \node[vertex] (G_1) at (0,0) {};
  \node[vertex] (G_2) at (-1,-1) {};
  \node[vertex] (G_3) at (1,-1) {};
  \node[vertex] (G_4) at (-3,-2) {};
  \node[vertex] (G_5) at (3,-2) {};
  \node[vertex] (G_6) at (-1,-2) {};
  \node[vertex] (G_7) at (1,-2) {};
  
  \node[vertex] (G_8) at (-3.5,-3) {};
  \node[vertex] (G_9) at (-2.5,-3) {};
  \node[vertex] (G_10) at (-1.5,-3) {};
  \node[vertex] (G_11) at (-0.5,-3) {};
  \node[vertex] (G_12) at (0.5,-3) {};
  \node[vertex] (G_13) at (1.5,-3) {};
  \node[vertex] (G_14) at (2.5,-3) {};
  \node[vertex] (G_15) at (3.5,-3) {};

  \draw [-, line width=0.3mm, black] (G_1) -- (G_2); 
  \draw [-, line width=0.3mm, black] (G_1) -- (G_3);
  \draw [-, line width=0.3mm, black] (G_2) -- (G_4);
  \draw [-, line width=0.3mm, black] (G_2) -- (G_6);
  \draw [-, line width=0.3mm, black] (G_3) -- (G_5);
  \draw [-, line width=0.3mm, black] (G_3) -- (G_7);
  
  \draw [-, line width=0.3mm, black] (G_4) -- (G_8); 
  \draw [-, line width=0.3mm, black] (G_4) -- (G_9); 
  \draw [-, line width=0.3mm, black] (G_6) -- (G_10); 
  \draw [-, line width=0.3mm, black] (G_6) -- (G_11); 
  \draw [-, line width=0.3mm, black] (G_7) -- (G_12); 
  \draw [-, line width=0.3mm, black] (G_7) -- (G_13); 
  \draw [-, line width=0.3mm, black] (G_5) -- (G_14); 
  \draw [-, line width=0.3mm, black] (G_5) -- (G_15); 
  
\draw [->, line width=0.3mm, blue] [dashed] (G_1) to[out=270,in=0] (G_2);
  
\draw [->, line width=0.3mm, blue] [dashed] (G_3) to[out=0,in=0] (G_1);
\draw [->, line width=0.3mm, blue] [dashed] (G_2) to[out=180,in=180] (G_1);

\draw [->, line width=0.3mm, blue] [dashed] (G_4) to[out=270,in=0] (G_8);
\draw [->, line width=0.3mm, blue] [dashed] (G_6) to[out=270,in=0] (G_10);
\draw [->, line width=0.3mm, blue] [dashed] (G_7) to[out=270,in=180] (G_13);
\draw [->, line width=0.3mm, blue] [dashed] (G_5) to[out=270,in=180] (G_15);

\draw [->, line width=0.3mm, blue] [dashed] (G_8) to[out=180,in=180] (G_4);
\draw [->, line width=0.3mm, blue] [dashed] (G_9) to[out=0,in=0] (G_4);
\draw [->, line width=0.3mm, blue] [dashed] (G_10) to[out=180,in=180] (G_6); 
\draw [->, line width=0.3mm, blue] [dashed] (G_11) to[out=0,in=0] (G_6); 
\draw [->, line width=0.3mm, blue] [dashed] (G_12) to[out=180,in=180] (G_7); 
\draw [->, line width=0.3mm, blue] [dashed] (G_13) to[out=0,in=0] (G_7); 
\draw [->, line width=0.3mm, blue] [dashed] (G_14) to[out=180,in=180] (G_5); 
\draw [->, line width=0.3mm, blue] [dashed] (G_15) to[out=0,in=0] (G_5); 

\end{tikzpicture}
\end{minipage}
\hspace{1cm}
\begin{minipage}[b]{.4\linewidth}
\begin{tikzpicture}
  \tikzstyle{vertex}=[circle,fill=black!25,minimum size=12pt,inner sep=2pt]
  
  \node[vertex] (G_1) at (0,0) {};
  
  \node[vertex] (G_2) at (-1,-1) {};
  \node[vertex] (G_3) at (1,-1) {};
  
  \node[vertex] (G_4) at (-2.5,-2) {};
  \node[vertex] (G_5) at (2.5,-2) {};
  \node[vertex] (G_6) at (-1,-2) {};
  \node[vertex] (G_7) at (1,-2) {};
  
  \node[vertex] (G_8) at (-2.5,-3) {};
  \node[vertex] (G_10) at (-1,-3) {};
  \node[vertex] (G_13) at (1,-3) {};
  \node[vertex] (G_15) at (2.5,-3) {};

  \draw [-, line width=0.3mm, black] (G_1) -- (G_2); 
  \draw [-, line width=0.3mm, black] (G_1) -- (G_3);
  \draw [-, line width=0.3mm, black] (G_2) -- (G_4);
  \draw [-, line width=0.3mm, black] (G_2) -- (G_6);
  \draw [-, line width=0.3mm, black] (G_3) -- (G_5);
  \draw [-, line width=0.3mm, black] (G_3) -- (G_7);
  
  \draw [-, line width=0.3mm, black] (G_4) -- (G_8); 
  \draw [-, line width=0.3mm, black] (G_6) -- (G_10); 
  \draw [-, line width=0.3mm, black] (G_7) -- (G_13); 
  \draw [-, line width=0.3mm, black] (G_5) -- (G_15); 
  
\draw [->, line width=0.3mm, blue] [dashed] (G_1) to[out=270,in=0] (G_2);
  
\draw [->, line width=0.3mm, blue] [dashed] (G_3) to[out=0,in=90] (G_5);
\draw [->, line width=0.3mm, blue] [dashed] (G_2) to[out=180,in=90] (G_4);

\draw [->, line width=0.3mm, blue] [dashed] (G_4) to[out=0,in=0] (G_8);
\draw [->, line width=0.3mm, blue] [dashed] (G_6) to[out=0,in=0] (G_10);
\draw [->, line width=0.3mm, blue] [dashed] (G_7) to[out=180,in=180] (G_13);
\draw [->, line width=0.3mm, blue] [dashed] (G_5) to[out=180,in=180] (G_15);
\draw [->, line width=0.3mm, blue] [dashed] (G_8) to[out=180,in=180] (G_4);
\draw [->, line width=0.3mm, blue] [dashed] (G_10) to[out=180,in=180] (G_6); 
\draw [->, line width=0.3mm, blue] [dashed] (G_13) to[out=0,in=0] (G_7); 
\draw [->, line width=0.3mm, blue] [dashed] (G_15) to[out=0,in=0] (G_5); 

\end{tikzpicture}
\end{minipage}
\vspace{-0.2cm}
\caption{Optimal backup placement in fully balanced (left) and "single leafs" (right) trees.} 
\label{fig7}
\end{figure}
\vspace{-0.8cm}
\clearpage

\newpage
\section{Proofs}

\subsection{Lemma 2.2 proof: Backup Placement in Forest Cover of Dense Graphs}
\label{appendix:forest}
\begin{proof}
As a first step, we need to compute a forest cover of any input graph $G$, such that the trees of the forest are {\em vertex-disjoint}. This can be done by extending the algorithm of Panconesi and Rizzi for $\Delta$-forests-decomposition \cite{panconesi2001some}. (Where $\Delta$ is the maximum degree in the input graph.) The latter algorithm produces edge-disjoint trees, that are not necessarily vertex-disjoint. Specifically, it partitions the edges of any graph into $\Delta$ acyclic subgraphs, such that each vertex has at most one parent in each subgraph. In other words, each of the $\Delta$ subgraphs is a forest in which vertices know their parents. The algorithm of Panconesi and Rizzi requires just a single round on any graph. An important property of the algorithm of \cite{panconesi2001some} we employ is that for any pair of neighbors in the input graph $G$, the higher-$ID$ neighbor is the parent of the lower-$ID$ neighbor in some forest of the output.

 Consider a subgraph $G = (V,E'), E' \subseteq E$, obtained as follows.
 First, invoke the $\Delta$-forest-decompo\-sition algorithm of \cite{panconesi2001some}, and initialize $E' = \emptyset$. Next, each vertex $v \in V$ that has at least one parent in at least one forest, selects an arbitrary parent $u$, and adds $(u,v)$ to $E'$. Note that the number of edges added this way is bounded by $|V| = n$, and the resulting graph $G'$ is acyclic. Indeed, a cycle would imply that either a vertex has selected two parents, or the cycle is consistently oriented, i.e., a vertex is an ancestor of itself. But each vertex selects at most one parent, and each parent has a higher $ID$ than that of its child, thus a vertex cannot be its own ancestor. Thus there are no cycles in $G'$. Note that the forest $G'$ is not necessarily a forest cover in this stage. This is because some vertices that are roots in all forests of Panconesi-Rizzi's execution may belong to none of the edges in $E'$. This can be corrected as follows. Each such root vertex $v$, that has not been selected by any of its children in the previous stage, selects an arbitrary neighbor $u$ and adds $(u,v)$ to $E'$. Note that such root vertices form an independent set, i.e., they cannot be neighbors in $G$. (Otherwise, one of them would be the parent of the other in some forest of Panconesi-Rizzi's execution.) Consequently, all edges added in this stage do not close cycles, since each such edge has one endpoint of degree $1$ in the resulting graph $G' = (V,E')$. Moreover, in this stage each vertex of $V$ belongs to at least one edge of $E'$, thus $G'$ is a forest cover. In addition, each vertex belongs to exactly one tree, thus the trees are vertex-disjoint. 
 \end{proof}



\subsection{Lemma 2.4 proof}
\label{o1}
\begin{proof}
The algorithm consists of two parts, each of which has a constant running time: (1) distributed selection of a backup placement node, in which each non-leaf selects one of its children, and (2) a distributed selection of a backup placement node, in which each leaf selects one of its neighbors. Therefore, the total running time is $O(1)$.
\end{proof}

\subsection{Lemma 2.5 proof}
\label{logn}
\begin{proof}
Initially, all vertices have unique $ID$s of size $O(\log n)$ bits each.
Computing a forest cover requires each vertex to know who of its neighbors have greater $ID$s than its own, and whether some neighbor has selected the vertex. This requires $O(\log n)$ bits per edge. Once the forest cover has formed,
each vertex can send its own $ID$ and its parent $ID$ to all of its neighbors. By that, each vertex learns if a neighbor is also one of its siblings in a tree $T$. (They are siblings if they have the same parent.) Also, using this technique each vertex can identify the sibling of a smaller $ID$ that is closest to its own.

This information is sufficient to select a neighbor for a backup placement by our algorithm. Sending two $ID$s over an edge requires $O(\log n)$ bits. Thus, the algorithm can be implemented with messages of size $O(\log n)$ per link per round.
\end{proof}

\subsection{UDG, BDG, and Line Graphs Neighborhood Independence}
\label{appendix:udg-bdg}


\begin{figure}[htb]
    \centering
    \begin{minipage}{0.45\textwidth}
        \centering

\usetikzlibrary{intersections, calc, fpu, decorations.pathreplacing}

\newcommand{\TikZFractionalCake}[5]{
    \pgfmathsetmacro{\angle}{360/#2};%
    \foreach [count=\x] \p in {1, 2, 3, 4, 5}%
    {   \pgfmathsetmacro{\lox}{\x-1}%
        \filldraw[draw=#4,fill=#3] [dashed] (0,0) -- (\angle*\lox:#5) arc (\angle*\lox:\angle*\x:#5) -- cycle;%
        \pgfmathsetmacro{\mix}{\x-0.5}%
        \node[rotate=\mix*\angle] at (\mix*\angle:#5*0.5+.3) {

        \begin{tikzpicture}

  		\tikzstyle{vertex}=[circle,fill=black!25,minimum size=12pt,inner sep=2pt]
        \tikzstyle{vertex2}=[circle,fill=white!25,minimum size=12pt,inner sep=2pt]
  
      \node[vertex] (G_0) at (-7.5,-1) {$v_{\x}$};
      \node[vertex2] (G_1) at (-6,-0.5) {};
      \node[vertex] (G_2) at (-9.6,-1) {$u$};
      \node[vertex2] (G_4) at (-6,-1.5) {};
      
      \draw [->, line width=0.2mm, blue] [dashed] (G_0) to (G_2);
      
		\end{tikzpicture}
        
        };
        \node[rotate=\mix*\angle] at (\mix*\angle:#5+0.3) {\x};
    }

}   

\begin{tikzpicture}
\TikZFractionalCake{5}{5}{white}{black}{3}
\end{tikzpicture}
\caption{Unit Disk Graph model: A root node $u$, which forms 5 different cliques, and $\{v_1, ..., v_5 \}$ nodes which must choose $u$ as their backup placement.}
\label{fig5}
\end{minipage}\hfill
    \begin{minipage}{0.5\textwidth}
        \centering
        
\usetikzlibrary{intersections, calc, fpu, decorations.pathreplacing}

\newcommand{\TikZFractionalCake}[5]{
    \pgfmathsetmacro{\angle}{360/#2};%
    \foreach [count=\x] \p in {1, 2, 3, 4, 5, 6, 7, 8, 9, 10, 11}%
    {   \pgfmathsetmacro{\lox}{\x-1}%
        \filldraw[draw=#4,fill=#3] [dashed] (0,0) -- (\angle*\lox:#5) arc (\angle*\lox:\angle*\x:#5) -- cycle;%
        \pgfmathsetmacro{\mix}{\x-0.5}%
        \node[rotate=\mix*\angle] at (\mix*\angle:#5*0.5+.3) {

        \begin{tikzpicture}

  		\tikzstyle{vertex}=[circle,fill=black!20,minimum size=12pt,inner sep=2pt]
        \tikzstyle{vertex2}=[circle,fill=white!25,minimum size=12pt,inner sep=2pt]
        \tikzstyle{vertex3}=[circle,fill=black!10,minimum size=1pt,inner sep=0.5pt]
  
      \node[vertex] (G_0) at (-7.5,-1) {$v_{k+\x}$};
      \node[vertex2] (G_1) at (-6,-0.5) {};
      \node[vertex] (G_2) at (-9.6,-1) {$u$};
      \node[vertex2] (G_4) at (-6,-1.5) {};
      
      \node[vertex3] (G_5) at (-8.5,-1) {$v_{\x}$};
      
      \draw [->, line width=0.2mm, blue] [dashed] (G_0) -- (G_2) node [pos=0.1, below, sloped] (TextNode) {\ldots};

		\end{tikzpicture}
        
        };
        \node[rotate=\mix*\angle] at (\mix*\angle:#5+0.3) {\x};
    }

}   

\begin{tikzpicture}
\TikZFractionalCake{11}{11}{white}{black}{3}
\end{tikzpicture}
        \caption{Bounded Disk Graph model: A root node $u$, which forms $11\times(\log{\frac{R_{max}}{R_{min}}})$ different cliques, such that $\{v_{1}, ..., v_{11} \}$ ... $\{v_{k+1}, ..., v_{k+11} \}$, $k=\log{\frac{R_{max}}{R_{min}}}$ nodes must choose $u$ as their backup placement.}
        \label{fig6}
    \end{minipage}
\end{figure}

\begin{lemma}
Let $G$ be a unit disk graph, the neighborhood independence of a vertex $u$ is at most $c=5$, i.e., every six neighbors of a vertex have at least one edge.
\end{lemma}

\begin{proof}
Let R be the radius of of the UDG (Figure \ref{fig5}). Let assume the neighborhood independence of a vertex $k$ is $c=6$, with $v_1, ..., v_6$ vertices set in a cyclic fashion around $k$. Because of the known UDG's features, $distance(k,v_i) \leq R$. If $(v_i,v_j) \notin E$, then $distance(v_i, v_j) > R$. Thus, $distance(v_i, v_j) > distance(k,v_i)$ \& $distance(k,v_j)$. Hence, in $\Delta$$kv_iv_j$, $(v_i,v_j)$ is the largest side, and so $\angle$$f$ is the largest angle, which must be $> \frac{\pi}{3}$. Therefore, $\sum_{i=1}^{6}\angle{v_ikv_j} > 6\times\frac{\pi}{3} = 2\pi$, a contradiction.
\end{proof}

\begin{lemma}
Let $G$ be a bounded disk graph, the neighborhood independence of a vertex $u$ is at most $c=11\times(\log{\frac{R_{max}}{R_{min}}})$, i.e., every $12\times(\log{\frac{R_{max}}{R_{min}}})$ neighbors of a vertex have at least one edge.
\end{lemma}

\begin{proof}

Let \(I(u)\) be the largest independent set in the subgraph induced by \({u} \cup N(u)\). We define a ?moat? \(b(i)\) for \(i \geq 0\), which is the annulus defined by circles of radii \(c^{i}R_{min}\) and \(c^{i+1}R_{min}\), centered at \(u\). Therefore we have \(n = \log_c{\frac{R_{max}}{R_{min}}}\) moats at max (\(c^nR_{min} = R_{max}\)). \(c=\sqrt[]{3}\) for ease of proof using geometric properties. Let \(N_i(u)\) be the be the neighbors of \(u\) that are in moat \(b(i)\). Let \(p,q\) be two vertices that are in \(N_i(u)\), and without lost of generality let \(R_p \leq R_q\).

The distance between \(p\) and \(q\) is at least \(min(R_p,R_q)=R_p\), since there is no edge between \(p\) and \(q\). We can "shrink" the circle centered at \(p\) with radius \(R_p\) until \(u\) is \textit{on the boundary} of the circle, \(R^{'}_p \leq R_p\) and \(R^{'}_p \geq c^{i}R_{min}\). The distance between \(q\) and \(u\) is at most \(c^{i+1}R_{min}\), which is \(\leq cR^{'}_p \). \(q\) is \textit{inside} the circle centered at \(u\) with radius \(cR^{'}_p\), but \textit{outside} the circle centered at \(p\) with radius \(R^{'}_p\) - this implies that \(q\) is in the crescent shaped shaded region.

In order to compute the angle $\alpha$ between $R^{'}_p$ and $c^{i+1}R_{min}$, we assume a triangle $\Delta R_p, cR^{'}_p, R^{'}_p$, as $c^{i+1}R_{min} \leq cR^{'}_p$ and $R_p < R^{'}_p$. Thus, $\arccos(\alpha) > (c^{2}{R^{'}_p}^{2}+{R^{'}_p}^{2}-{R^{'}_p}^{2}) / (2c{R^{'}_p}^{2}) = c^{2}/2c = \frac{\pi}{6}^{\circ}$ .

Under these circumstances, the angle between \(p\) and \(q\) at \(u\) (\(\alpha\)) is \(> \frac{\pi}{6}\), and therefore there cannot be more than 11 vertices in the moat \(b(i)\) \(\rightarrow\) $11\times(\log{\frac{R_{max}}{R_{min}}})$ (Figure \ref{fig6}). We proclaim $\log{\frac{R_{max}}{R_{min}}}$ is assumed in this article to be a constant as wireless networks tend to work in a relatively small $\log{\frac{R_{max}}{R_{min}}}$ proportions.

\end{proof}

\begin{lemma} \label{lem:indln}
Let $G = (V,E)$ be a graph, and let $L(G)$ be a line graph of G. The neighborhood independence of a vertex $u$ in $L(G)$ is at most $c=2$.
\end{lemma}

\begin{proof}
A vertex $v$ in $L(G)$ corresponds to an edge $e_v \in G$. An independent set $I$ of neighbors of $v$ corresponds to a set of edges $J_I \subseteq E$, such that each $e \in J_I$ share a common endpoint with $e_v$, but each
pair of edges $e, e' \in J_I$ do not intersect. Therefore, $|J_I| \leq 2$, and consequently, $|I| \leq 2$. Thus, any line graph $L(G)$ has neighborhood independence bounded by 2.
\end{proof}


\section{Self-stabilizing Backup Placement in Graphs of Bounded Neighborhood Independence}
\label{selfstab}
In this section we devise a self-stabilizing backup placement algorithm in Dijkstra model of self-stabilization \cite{edsger1974dijkstra}. In this model each vertex has a ROM (Read Only Memory) that is failure free, and a RAM (Random Acess Memory) that is failure prone. An adversary can corrupt the RAM of all processors in any way. However, in certain periods of time, faults do not occur. These periods of time are not known to the processors. The goal of a distributed self-stabilizing algorithm is reaching a proper state in all processors, once faults stop occurring. Since these time points are not known, an algorithm is constantly executed by all processors. The stabilization time is the number of rounds from the beginning of a time period in which faults do not occur, until all processors reach a proper state, given that no additional faults occur during this time period.

Our algorithm stores only the $ID$ of a processor in its ROM. The RAM of each processor $v \in V$ stores its parent in the forest and its neighbor for a backup. These are stored in the variables $v.Parent$, $v.BP$, respectively. In each round, the forest is re-computed, and the backup placement is performed according to the rules of our algorithm, described in Section \ref{intro2}. The pseudocode of the self-stabilizing algorithm is provided in Algorithm \ref{algo5} below.

\begin{algorithm}[htb]
\caption{Self-Stabilizing Backup Placement}
\label{algo5}
\begin{algorithmic}[1]
\State {\bf foreach} node $v \in V$ in each round in parallel do:
\Procedure{Stab-BP($G = (V,E)$)}{}
\State send $v.ID$, $v.Parent$, $v.BP$ to all neighbors
\State receive all messages of the current round
\State {\bf /* Forest construction */}
\If{$v$ has a neighbor in $G$ with greater $ID$ than its own}
	\State $v.Parent \gets$ a neighbor of $v$ in $G$ with a greater $ID$ that is closest to its own.
\EndIf
\If{$v$ has no neighbor in $G$ with a greater $ID$ {\bf and} there exists a neighbor $u$ of $v$ in $G$ with $u.Parent = v$}
  \State $v.Parent \gets$ NULL
\EndIf
\If{$v$ has no neighbor in $G$ with a greater $ID$ {\bf and} there is no neighbor $u$ of $v$ in $G$ with $u.Parent = v$}
  \State $v.Parent \gets$ a neighbor in $G$ with the closest $ID$ to its own.
\EndIf 
\State {\bf /* Backup placement computation */}
\If{$v$ has a neighbor $u$ in $G$ with $u.parent = v$}
  \State $v.BP \gets$ such a neighbor $u$, where ties are broken by taking the smallest $ID$
\ElsIf {there is a neighbor $u$ of $v$ in $G$ with $u.parent = v.parent$ {\bf and} $u.ID < v.ID$}
   \State $v.BP \gets$ such a neighbor $u$ whose $ID$ is closest to that of $v$.
\Else \State $v.BP \gets v.Parent$
\EndIf
\EndProcedure
\end{algorithmic}
\end{algorithm}

Next, we prove that the algorithm stabilizes within three rounds, once faults stop occurring. First, we show that each vertex belongs to the forest, within two rounds from that moment. 

\begin{lemma}
Within two rounds from the moment when faults stop occurring, the set of variables $v.Parent$ of all $v \in V$ form a forest.
\end{lemma}
\begin{proof}
In the first round, once faults stop occurring, each vertex whose $ID$ is not a local maximum selects a neighbor with a greater $ID$. Thus, the set of all such vertices form a forest in the end of this round. Since this step depends solely on information in the ROM, this step outcome is going to be the same in each faultless round. For the next step, denote the set of vertices with locally maximum $ID$s by $W$. Note that $W$ is an independent set, since otherwise there is a pair of neighbors in $W$. But neighbors cannot both have local maximum $ID$s. In the second round from the moment that faults stop occurring, each vertex in $W$ discovers whether it belongs to the forest. If so, it marks itself as a root without a parent. Otherwise, it selects a parents in the forest. Consequently, after two rounds, all vertices belong to the forest. Note that in the end of the first round, though the selections of $V \setminus W$ produce no cycles, there may be still cycles due to earlier faults of vertices of $W$. However, this is corrected in round $2$, after which there are no cycles at all. 
\end{proof}
In the following rounds, the selections of parents do not change for all $v \in V$, as long as there are no faults. Indeed, the forest construction steps result in the same outcome of each faultless round, after the second round. Hence in the third round a proper backup placement is formed.  Moreover, since it is computed using a forest of a graph with bounded neighborhood independence, according to the rules described in Section \ref{intro2}, we obtain the following result.
\begin{theorem}
In graphs with neighborhood independence bounded by a constant, our algorithm stabilizes within $3$ rounds and produces $O(1)$-backup-placement.
\end{theorem}

\end{appendices}

\end{document}